\documentclass[12pt]{article}
\usepackage{amsmath}
\usepackage{amssymb,amsthm}
\usepackage{color}

\newtheorem{Thm}{Theorem}[section]
\newtheorem{theorem}[Thm]{Theorem}
\newtheorem{proposition}[Thm]{Proposition}
\newtheorem{corollary}[Thm]{Corollary}

\newtheorem{remark}{Remark}[section]

\newtheorem{definition}[Thm]{Definition}

\title{A criterion to characterize interacting theories in the Wightman framework}

\author{Christian D.\ J\"akel\footnote{jaekel@ime.usp.br, Dept.~de Matem\'atica Aplicada, 
Univ. de S\~ao Paulo (USP), Brasil} \ 
and Walter F. Wreszinski\footnote{wreszins@gmail.com, 
Instituto de Fisica, Universidade de S\~ao Paulo (USP), Brasil}}        

\begin{document}

\maketitle

\begin{abstract}
We propose a criterion to characterize interacting theories in a suitable Wightman framework
of relativistic quantum field theories which 
incorporates a ``singularity hypothesis'', which has been conjectured 
for a long time, is supported by renormalization group theory,
but has never been formulated mathematically. The (nonperturbative) wave function
renormalization $Z$ occurring in these theories is shown not to be necessarily equal
to zero, except if the equal time commutation relations (ETCR) are assumed. Since the ETCR
are not justified in general (because the interacting fields cannot in general be restricted to sharp times,
as is known from model studies), the condition $Z=0$ is not of general validity in interacting 
theories. We conjecture that it characterizes either unstable (composite) particles or the   
charge-carrying particles, which become infraparticles in the presence of massless particles.
In the case of QED, such ``dressed'' electrons are not expected to be confined, but
in QCD we propose a quark confinement criterion, which follows naturally from lines suggested by 
the works of Casher, Kogut and Susskind and Lowenstein and Swieca.
 
\end{abstract}

\section{Introduction}

In his recent recollections, 'tHooft (\cite{tHooft}, Sect.~5) emphasizes that 
an asymptotic (divergent) series, such as the power series for the scattering (S) 
matrix in the coupling constant $\alpha = \frac{1}{137}$ in quantum 
electrodynamics ($qed_{1+3}$) does not define a theory rigorously (or mathematically). 
In the same token, Feynman was worried about whether the qed S matrix would be ``unitary 
at order 137'' (\cite{Wight2}, discussion  on  p.~126), and in Section 5 of \cite{tHooft}, 'tHooft 
remarks that the uncertainties in the S matrix amplitudes at order~137 for qed are comparable 
to those associated to the ``Landau ghost'' (or pole) \cite{Landau}.  

The main reason to believe in quantum field theory  is, therefore, 
in spite of the spectacular success of perturbation theory (\cite{Weinb1}, 
Chap.~11 and \cite{Weinb3}
Chap.~15), strongly tied to non-perturbative approaches. 

In his famous Erice lectures of 1979 \cite{Wight}, ''Should we believe in quantum field theory?'',
Wightman remarks (p.~1011) that he
once expressed to Landau his lack of confidence in the arguments 
which he and his co-workers had
put forward for the inconsistency of field theory. 
``He then offered me the following: You agree
that the essential problem of quantum field theory is its high-energy 
behavior? Yes. You agree that
up to now no-one has suggested a consistent high-energy behavior 
for quantum field theory? Yes. 
Then you have to believe in the inconsistency of quantum field theory, 
because physicists are smart 
and if there was a consistent high-energy behavior, they would have found it!''

Landau's remarks concern, in particular, quantum 
electrodynamics in the Coulomb gauge in three space dimensions - $qed_{1+3}$. 
He refers, as we do
throughout the paper, to ``bona-fide'' field theories in which cutoffs have 
been eliminated,
and are thus invariant under certain symmetry groups: ``effective'' field 
theories are thereby excluded. 
Our recent result \cite{JaWre2}, however, establishes positivity of the (renormalized) 
energy, uniformly in the volume ($V$) and ultraviolet ($\Lambda$) cutoffs. This stability 
result may be an indication of the absence of Landau poles or ghosts in $qed_{1+3}$ . 
In order to achieve this, the theory in Fock space (for 
fixed values of the cutoffs)  is exchanged for a formulation in which 
(in the words of Lieb and Loss, who were the first to exhibit this phenomenon in a 
relativistic model \cite{LLrel}), ``the electron Hilbert space is linked 
to the photon Hilbert space in an inextricable way''. Thereby, ``dressed photons'' 
and ``dressed electrons'' arise as new entities. 

The above picture, required by stability, provides a physical characterization of the
otherwise only mathematically motivated non-Fock representations which arise in an
interacting theory (\cite{Wight},\cite{Wight1}). The non-unitary character of the transformations 
to the physical Hilbert space,
when the space and ultraviolet cutoffs are removed, are intrinsic to the singular nature of field 
theory, and never really belonged to the usual lore of (even a smart!) theoretical physicist
(for an exception, see the book by G. Barton \cite{Barton1}), which might explain Landau's remarks,
as well as the fact that the answer to the question posed in the title involves considerations
both of foundational and mathematical nature.

In this paper,  we attempt to 
incorporate the findings of \cite{JaWre2} in a proper general 
framework. For this purpose, we suggest a criterion characterizing interacting theories, 
incorporating a ``singularity hypothesis'' which has been conjectured for a long time, 
is supported by renormalization group theory,
but has never been formulated mathematically. We do so 
in Definition~3.3. 
With this hypothesis, it is 
possible to prove Theorem~\ref{thm:3}, which is our main result. It permits 
to characterize
interacting quantum field theories, in a suitable Wightman framework, by the
property that the total spectral measure is infinite.
The (non-perturbative) wave-function renormalization occurring in these theories,
defined in Proposition~\ref{prop:1}, is shown not to be universally equal to
zero in interacting theories satisfying the singularity hypothesis,
except if equal time commutation relations (ETCR) are assumed (Corollaries ~\ref{cor:2}
and ~\ref{cor:3}).
Since the ETCR are known, from model studies, not to be universally valid,
the physical interpretation of the condition $Z=0$ is open to question.
We suggest that it may characterize either the existence of ``dressed'' particles in charged 
sectors, which may occur for specific theories 
such as $qed_{1+3}$) due to the presence of massless particles (the photon), by a theorem of 
Buchholz \cite{Buch}, or unstable (composite)
particles. The latter conjecture, due to Weinberg (\cite{Weinb1}, p. 460), is, however, 
still unsupported by any rigorous result.

In the case of $qed_{1+3}$ , such ``dressed'' electrons are not expected to be confined, but
in qcd we propose a quark confinement criterion, which follows naturally from lines suggested by 
by the works of Casher, Kogut and Susskind \cite{CaKoSu}, and 
Lowenstein and Swieca \cite{LSwi} 
in massless $qed_{1+1}$ (the Schwinger model),
in which the ``dressing'' of the electrons by photons is rather drastic, 
the electron field being 
expressed entirely as a functional of the photon field, which 
becomes massive.
 
\section{The Field Algebra}

In order to formulate the above-mentioned phenomena in $qed_{1+3}$, we consider a 
\emph{field algebra} ${\cal F}^{0} \equiv \{A_{\mu}^{0},\psi^{0},\bar{\psi^{0}}\}$ (containing 
the identity) generated by
the free vector potential $A_{\mu}^{0}, \mu=0,1,2,3$, and the 
electron-positron fields~$\psi^{0},\bar{\psi^{0}}$. 
We may assume that the field algebra is initially defined 
on the Fock-Krein (in general, indefinite-metric, see \cite{Bognar}) tensor product of 
photon and fermion Fock spaces. However, 
of primary concern for us will be an \emph{inequivalent} 
representation of the field algebra  ${\cal F}^{0}$ on a (physical) 
Hilbert space ${\cal H}$, with generator of time-translations - the physical 
Hamiltonian $H$ - satisfying \emph{positivity}, \emph{i.e.}, 
	\begin{equation}
		H \ge 0 \; ,
		\label{(1)}
	\end{equation}
and such that
	\begin{equation}
		H \Omega = 0  \; ,
		\label{(2)}
	\end{equation}
where $\Omega \in {\cal H}$ is the vacuum vector, and
	\begin{align}
		A_{\mu} &\equiv A_{\mu}(A_{\mu}^{0},\psi^{0}, \bar{\psi^{0}})  \; ,
		\label{(3)}
		\\
		\Psi & \equiv \Psi(A_{\mu}^{0},\psi^{0}, \bar{\psi^{0}})  \; ,
		\label{(4.1)}
		\\
		\bar{\Psi} & \equiv \bar{\Psi}(A_{\mu}^{0},\psi^{0}, \bar{\psi^{0}})  \; .
		\label{(4.2)}
	\end{align}  

We may define, as usual, the (restricted) class of c-number $U(1)$ local 
gauge transformations, acting on ${\cal F}^{0}$, by the maps
	\begin{align}
		A_{\mu}^{0}(f) & \to A_{\mu}^{0}(f) +c \, \langle f,\partial^{\mu}u \rangle  \; ,
		\label{(5.1)}
		\\
		\psi^{0}(f) & \to \psi^{0} \bigl( {\rm e}^{iu} f \bigr)  \; ,
		\label{(5.2)}
		\\
		\bar{\psi^{0}}(f) & \to \bar{\psi^{0}} \bigl( {\rm e}^{-iu} f \bigr)  \; ,
		\qquad f \in {\cal S}(\mathbb{R}^{1+s}) \; . 
		\label{(5.3)}
	\end{align}
In qed, $c=\frac{1}{e}$ and $u$ satisfies certain regularity conditions, which guarantee that 
	\[
		{\rm e}^{\pm iu}f \in {\cal S}(\mathbb{R}^{1+s}) \quad \text{if} \; f \in {\cal S}(\mathbb{R}^{1+s}) \; , 
		\quad \text{and} \quad  | \langle f,\partial^{\mu}u \rangle | < \infty \; . 
	\]
In the following, ${\cal S}$ denotes  Schwartz space (see, \emph{e.g.}, \cite{BB}), 
$\langle \, . \, , \, . \, \rangle$  denotes  the $L^{2}(\mathbb{R}^{1+s})$ scalar product 
and $s$ is the space dimension. Such transformations have been considered in a quantum context, that of (massless) relativistic qed in 
two space-time dimensions (the Schwinger model) by Raina and Wanders,
but their unitary implementability is a delicate matter \cite{RaWa}. 
We shall use (6)--(8) merely as
as a guiding principle to construct the observable algebra, to which 
we now turn. 

The \emph{observable algebra} is assumed to consist of 
gauge-invariant objects, namely the tensor fields
	\begin{equation}
		F_{\mu,\nu} = \partial_{\mu} A_{\nu}-\partial_{\nu}A_{\mu}  \; ,
		\label{(6)}
	\end{equation}
with $A_{\mu}$ given by (3), describing the dressed photons, and the quantities (10)
below. In order to define them, we assume the existence of gauge-invariant
quantities $\Psi,\bar{\Psi}$ in \eqref{(4.1)}, \eqref{(4.2)}, which create-destroy 
electrons-positrons ``with their photon clouds''. 

While we hope that the results in \cite{JaWre2} will eventually lead to an (implicit or explicit) 
expression for $\Psi,\bar{\Psi}$, it should be emphasized that this is a very difficult, 
open problem; see the important work of Steinmann in perturbation theory \cite{Steinmann1}. 
We also note that when we  consider the \emph{vacuum sector},
the fermion part of the observable algebra will be assumed to consist of the combinations
	\begin{align}
		A(f,g) & \equiv \bar{\Psi}(f)\Psi(g) \mbox{ with } f,g \in {\cal S}(\mathbb{R}^{1+s})  \; ,
		\nonumber
		\\
		B(f,g) & \equiv \Psi(f)\bar{\Psi}(g) \mbox{ with } f,g \in {\cal S}(\mathbb{R}^{1+s})  \; ,
		\label{(7)}
	\end{align}
The quantities $A$ and $B$ above stand for rigorous versions of  quantities of the form \eqref{(8a)}, smeared with 
functions $f,g$, which have never been constructed except in the Schwinger model. The existence of \emph{charged sectors} is a related 
problem, which will concern us in Section 4.  

\section{A framework for relativistic quantum gauge theories}

In this paper we propose a framework which is not new, having been used by Lowenstein and 
Swieca \cite{LSwi} and Raina and Wanders \cite{RaWa} to construct a theory of $qed_{1+1}$, the
Schwinger model.

The theory will be defined by its \emph{n-point Wightman functions} \cite{StreWight} of 
observable fields. Alternatively, a Haag-Kastler theory~\cite{HK} 
may be envisaged. It has been shown in the seminal work of the latter authors that 
the whole content of a theory can be expressed in terms of its observable algebra. 
In the case of gauge theories, the latter corresponds to the algebra generated by gauge 
invariant quantities \cite{StrWight}. As remarked by Lowenstein and Swieca \cite{LSwi}, 
the observable algebra, being gauge-invariant, should have the same representations, 
independently of the gauge of the field algebra it is constructed from. Thus, $n$-point 
functions constructed over the observable algebra should be the same in all gauges. 
These remarks fully justify the  usage  
of non-covariant gauges, which, as we shall see, are of particular importance in a 
non-perturbative framework. For scattering theory 
and particle concepts within a theory of local observables, see \cite{ArHa}, \cite{BuPoSt}, 
and \cite{BuchSum} for a lucid review.

There exist various arguments supporting the use of non-covariant gauges in 
relativistic quantum field theory: they are of both physical and mathematical nature. 
In part one of his treatise, Weinberg notes (\cite{Weinb1}, p.~375, Ref.~2): ``the use of 
Coulomb gauge in electrodynamics was strongly advocated by Schwinger  
on pretty much the same grounds as here: that we ought not to introduce photons with 
helicities other than $\pm 1$''. Indeed, as shown by Strocchi \cite{Strocchi}, a framework 
excluding ``ghosts'' necessarily requires the use of non-manifestly covariant gauges, 
such as the Coulomb gauge in $qed_{1+3}$, the Weinberg or unitary gauge in the Abelian 
Higgs model \cite{Weinb2},  and the Dirac \cite{Dirac} or light-cone gauge in quantum 
chromodynamics (\cite{SriBro}, \cite{Cornwall}). Another instance of the physical-mathematical advantage 
of a non-covariant gauge is the ``$\alpha=\sqrt{\pi}$'' gauge in massless $qed_{1+1}$ - the 
Schwinger model \cite{Schw} - see \cite{LSwi}, \cite{RaWa}. As in the Coulomb gauge 
in $qed_{1+3}$, there is no need for indefinite metric in this gauge, \emph{i.e.}, the zero-mass 
longitudinal part of the current is gauged away, and one has a solution of Maxwell's 
equations (as an operator-valued, distributional identity)
	\begin{equation}
		\partial_{\nu}F^{\mu,\nu}(x) = -ej^{\mu}(x)
		\label{(8)}
\end{equation}
on the whole Hilbert space. This is an important ingredient 
in Buchholz's theorem \cite{Buch}, to which we come back in the sequel. 

The structure of the observable algebra is quite simple in  the Coulomb 
gauge: the field \eqref{(6)} is 
just the electric field, which is defined in terms of a massive scalar field, the 
quantities \eqref{(7)} are, in this gauge, rigorous versions of the (path-dependent) quantities
	\begin{equation}
		\psi(x) {\rm e}^{ ie \int_{x}^{y} {\rm d} t_{\mu} \; A^{\mu}(t)} \psi^{*}(y)
                \label{(8a)}
        \end{equation}
and  their adjoints (in the distributional sense), see \cite{LSwi} and \cite{RaWa}. 
In the case of $qed_{1+3}$, such quantities are plagued by infrared divergences, 
see the discussion in \cite{Steinmann1}. As a consequence of  the simple structure of 
the observable algebra, one arrives at a correct physical-mathematical picture 
of spontaneous symmetry breakdown (\cite{LSwi},\cite{RaWa}); in covariant 
gauges this picture is masked by the presence of spurious gauge excitations.

In \cite{StreWight}, pp.~107-110, it is shown that if the (n-point) \emph{Wightman functions} 
satisfy 
\begin{itemize}
\item [$a.)$] the relativistic transformation law; 
\item [$b.)$] the spectral condition; 
\item [$c.)$] hermiticity; 
\item [$d.)$] local commutativity; 
\item [$e.)$] positive-definiteness, 
\end{itemize}
then they are the \emph{vacuum expectation values of a field theory} satisfying 
the Wightman axioms, except, eventually,  the uniqueness of 
the vacuum  state. We refer to \cite{StreWight} or \cite{RSII}  for 
an account of Wightman theory, and for the description of these properties.  It has been shown
in \cite{LSwi}, \cite{RaWa} that 
$qed_{1+1}$ in the ``$\alpha=\sqrt{\pi}$'' gauge satisfies $a.)-e.)$. 
The crucial positive-definiteness condition~e.) has been 
shown in \cite{LSwi} to be a consequence of the positive-definiteness of a subclass 
of the n-point functions of the Thirring model~\cite{Thirrm} in the formulation of 
Klaiber~\cite{Klai}. Positive-definiteness of the Klaiber n-point functions was 
rigorously proved by Carey, Ruijsenaars and Wright \cite{CRW}. Uniqueness of the 
vacuum holds in each irreducible subspace of the (physical) Hilbert 
space~${\cal H}$~\cite{RaWa}, as a result of the cluster property; see also \cite{LSwi}.

We shall \emph{assume} $a.)-e.)$ for the 
$n$-point functions of the observable  fields, with, in addition, the following requirement
\begin{itemize}
\item [$f.)$] interacting fields are assumed to satisfy the singularity hypothesis (the forthcoming  
Definition 3.3).
\end{itemize}

The crucial mathematical reason for choosing a non-covariant gauge is, as we shall see, the positive-definiteness 
condition e.). Concerning the uniqueness of the vacuum, we shall assume it is valid by 
restriction to an irreducible component of ${\cal H}$, as in $qed_{1+1}$.

\subsection{The K\"all\'en-Lehmann representation}

A major dynamical issue in quantum field theory is the (LSZ or Araki-Haag) asymptotic 
condition (see, \emph{e.g.}, \cite{BuchSum} and references given there), which 
relates the theory, whose objects are the fundamental observable fields, to \emph{particles}, 
described by physical parameters (\emph{mass} and \emph{charge} in qed). This issue is equivalent 
to the renormalization (or normalization) of perturbative quantum field theory, which 
itself is related to the construction of continuous linear extensions of certain functionals, 
such as to yield well-defined Schwartz distributions (see \cite{Scharf} and references 
given there). On the other hand, in a non-perturbative framework, a theory of renormalization 
of masses and fields also exists, and ``has nothing directly to do with the presence of 
infinities'' (\cite{Weinb1}, p.~441, Sect.~10.3). We adopt a related proposal, which we formulate 
here, for simplicity, for a theory of a self-interacting scalar field $A$ of mass $m$ satisfying 
the Wightman axioms (modifications are mentioned in the sequel).  
We assume that  $A$ is an operator-valued 
tempered distribution on the Schwartz space ${\cal S}$ (see \cite{RSII}, Ch.~IX). 

\bigskip
We have the following  result, concerning the spectral representation of the 
two-point function $W_{2}$ (\cite{RSII}, p.~70, Theorem IX-34):

\begin{theorem}[The K\"{a}ll\'{e}n-Lehmann representation]
\label{th:2} 
	\begin{equation}
		W_{2}^{m}(x-y)= \langle \Omega,A(x)A(y)\Omega \rangle
		= \frac{1}{i} \int_{0}^{\infty} {\rm d}\rho(m_\circ^{2}) \; 
		\Delta_{+}^{m_\circ}(x-y) \; , 
		\label{(9)}
\end{equation}
where $\Omega$ denotes the vacuum vector, $x=(x_{0},\vec{x})$, and
	\begin{align}
		\Delta_{+}^{m_\circ}(x) =\frac{i}{2(2\pi)^{3}} 
		\int_{\mathbb{R}^{3}} {\rm d}^{3}\vec{k} \; 
		\frac{ {\rm e}^{-ix_{0}\sqrt{m_\circ^{2}+\vec{k}^{2}}
		+i\vec{x}\cdot\vec{k}}}{\sqrt{m_\circ^{2}+\vec{k}^{2}}} 
		\label{(10)}
	\end{align}
is the two-point function of the free scalar field of mass $m_\circ$, and $\rho$ is 
a polynomially-bounded measure on $[0,\infty)$, \emph{i.e.},
	\begin{equation}
		\int_{0}^{L} {\rm d} \rho(m_\circ^{2}) \le C(1+L^{N})
		\label{(11)}
\end{equation}
for some constants $C$ and $N$. It is further assumed that
	\begin{equation}
		\langle \Omega,A(f)\Omega \rangle = 0 \qquad \forall  f \in {\cal S} \; . 
		\label{(12)}
\end{equation}
\end{theorem}

Note that \eqref{(9)} is symbolic; for its proper meaning, which relies on \eqref{(11)},
see \cite{RSII}. In the next proposition, since there is only
a finite number of masses in nature, we assume a priori that the pure point part (which
could eventually include a dense point spectrum, i.e., accumulation points) is, in fact, discrete,
containing only a finite number of mass values. We further assume that this discrete part of the 
measure is associated with the (renormalized) physical masses existent in nature, \emph{i.e.}, just one for 
a scalar field (and latter the electron and the photon mass in qed). This assumption is verified
for the free fields (scalar, spinor, vector), with $Z=1$.

It should be remarked that a scalar field of mass $m$ creates, with non-vanishing probability,
the state of a particle of mass $m$ \emph{in the sense of Wigner} (see, e.g., \cite{Weinb1})
from the vacuum. In particular, the electron in qed is not a particle, although the photon is 
(expected to be) a particle. It is nevertheless true that stability in \cite{JaWre2} and
\cite{LLrel} is achieved by the introduction of ``dressed photons''.

\begin{proposition} For a scalar field of mass $m \ge 0$, the 
measure ${\rm d}\rho(m_\circ^{2}) $ 
appearing in the K\"{a}ll\'{e}n-Lehmann spectral representation 
allows a decomposition
\label{prop:1} 
	\begin{equation}
		{\rm d}\rho(m_\circ^{2}) = Z\delta(m_\circ^{2}-m^{2}) + d\sigma(m_\circ^{2}) \; , 
		\label{(13)}
\end{equation}
where
	\begin{equation}
		0 < Z < \infty
		\label{(14)}
\end{equation}
and
	\begin{equation}
		\int_{0}^{L} {\rm d} \sigma(m_\circ^{2}) \le C_{1}(1+L^{N_{1}})
		\label{(15)}
\end{equation}
for some constants $C_{1}$ and $N_{1}$. 
\end{proposition} 

\begin{proof} 
By the Lebesgue decomposition (see, \emph{e.g.}, Theorems I.13, I.14, p.~22 of \cite{RSI})
	\begin{equation}
		{\rm d}\rho = {\rm d}\rho_{p.p.} + {\rm d}\rho_{s.c.} + {\rm d}\rho_{a.c.} \; , 
		\label{(16.1)}
\end{equation}
where p.p.~denotes the \emph{pure point}, s.c.~ denotes 
 the \emph{singular continuous} and 
a.c.~ denotes 
 the \emph{absolutely continuous} parts of ${\rm d}\rho$. 
By the assumptions, the pure point part of the measure is, in fact, discrete and, for 
a scalar field of mass $m$, we obtain
         \begin{equation}
		{\rm d}\rho_{p.p.}(m_\circ^{2}) = Z\delta(m_\circ^{2}-m^{2}) \; ,
		\label{(16.2)}
\end{equation}
where $Z$ satisfies \eqref{(14)}, and, by \eqref{(11)}, 
\eqref{(13)} and \eqref{(16.1)}, ${\rm d}\sigma$ satisfies \eqref{(15)}. 
\end{proof}

\begin{remark}
\label{Remark 0.1}
Of course, $Z=0$ in \eqref{(16.2)}, if there is no discrete component of mass $m$ in the total 
mass spectrum of the theory. In general $Z$ in each discrete component of $d\rho$ has only to satisfy
         \begin{equation}
               0 \le Z < \infty
               \label{(14.1)} \; , 
\end{equation}  
because of the positive-definiteness conditions e.) (or the positive-definite 
Hilbert space metric).
\end{remark}
 
\begin{remark}
\label{Remark 0.2}
Expression \eqref{(16.2)} corresponds precisely to (\cite{Weinb1}, p.~461, Equ.~(10.7.20)). 
Thus, Proposition~\ref{prop:1} is just a mathematical statement of the nonperturbative 
renormalization theory, as formulated by Weinberg. Thus, 
the physical interpretation of $Z$ is that $0< Z < \infty$ 
is the wave function renormalization constant, due to the 
fact that the physical field $A_{phys}$ is normalized (or renormalized) by the 
one-particle condition (\cite{Weinb1}, (10.3.6)) which stems from the LSZ (or 
Haag-Ruelle) asymptotic condition (see also Sections~2 and~5 of \cite{BuchSum} 
and references given there for the appropriate assumptions). A general field, 
as considered in \eqref{(9)}, does not have this normalization. By the same token, 
the quantity $m^{2}$ in \eqref{(16.2)} is interpreted as the physical (or renormalized) 
mass associated to the scalar field.  
\end{remark}

For $F_{\mu,\nu}$~and $\Psi$, we have the analogues of \eqref{(9)}, namely 
	\begin{align}
		& \langle F_{\mu,\nu}(x) \Omega, F_{\mu,\nu}(y)\Omega \rangle 
		= \int {\rm d} \rho_{ph}(m_\circ^{2})\int\frac{{\rm d}^{3}p}{2p_{0}} 
		(-p_{\mu}^{2}g_{\nu\nu}-p_{\nu}^{2}g_{\mu\mu}) {\rm e}^{ip \cdot (x-y)} \; , 
		\label{(19.1)}
	\end{align}
with $\mu \ne \nu$, no summations involved, and $p_{0}=\sqrt{ {\vec{p}\,}^{2}+m_\circ^{2}}$. 
Denoting spinor indices by $\alpha,\beta$, we have
	\begin{align}
		S_{\alpha,\beta}^{+}(x-y) 
		& = \langle \Omega, \Psi_{\alpha}(x)\bar{\Psi}_{\beta}(y) \Omega \rangle
		\label{(20.1)} \\
		& = \int_{0}^{\infty} {\rm d}\rho_{1}(m_\circ^{2}) S_{\alpha,\beta}^{+}(x-y;m_\circ^{2}) 
		+ \delta_{\alpha,\beta} \int_{0}^{\infty} {\rm d} \rho_{2}(m_\circ^{2}) 
		\Delta^{+}(x-y;m_\circ^{2}) \, , 
		\nonumber
	\end{align}
with ${\rm d} \rho_{ph}, {\rm d} \rho_{1}, {\rm d} \rho_{2}$ positive, polynomially bounded 
measures, and $\rho_{1}$ satisfying certain bounds with respect 
to $\rho_{2}$ (see \cite{Lehmann}, p.~350 for the notation). Again, as in \eqref{(16.2)},
	\begin{equation}
		{\rm d} \rho_{ph}(m_\circ^{2}) = Z_{3} \delta (m_\circ^{2}) 
		+ {\rm d} \sigma_{ph}(m_\circ^{2})
		\label{(19.2)}
	\end{equation}
and
	\begin{equation}
		{\rm d} \rho_{1}(m_\circ^{2}) = Z_{2}\delta(m_\circ^{2}-m_{e}^{2}) 
		+ {\rm d} \sigma_{1}(m_\circ^{2}) \; , 
		\label{(20.2)}
	\end{equation}
with $m_{e}$ the renormalized electron mass, according to conventional notation. 

\goodbreak 
We have, by the general condition \eqref{(14.1)},
	\begin{align}
		0 \le Z_{3} < \infty \; , 
		\label{(19.3)}
		\\
		0 \le Z_{2} < \infty \; . 
		\label{(20.3)}
	\end{align}
When $Z_{3} > 0$, the renormalized electron charge follows from 
(\cite{Weinb1}, (10.4.18)). Assumption \eqref{(12)}, which is also expected 
to be generally true on physical grounds, becomes
	\begin{align}
		\langle \Omega, F_{\mu,\nu}(f) \Omega \rangle & = 0 
		\quad \forall f \in {\cal S} \; , 
		\label{(19.4)}
		\\
		\langle \Omega, \Psi_{\alpha}(f) \Omega \rangle & = 0 \quad 
		\forall  f \in {\cal S} \mbox{ and } \alpha 
		\text{  a spinor index} \, . 
		\label{(20.4)}
	\end{align}

In summary, Proposition~\ref{prop:1} provides a rigorous (non-perturbative) definition 
of the wave-function renormalization constant. In the proof 
of Theorem~\ref{th:2} (\cite{RSII}, p.~70), the positive-definiteness condition 
e.) plays a major role. Thus, the definition of $Z$ and its range \eqref{(14.1)} (which 
depends on the positivity of the measure ${\rm d}\rho$) strongly hinge on the fact that 
the underlying Hilbert space has a positive metric. Parenthetically, the 
positive-definiteness condition on the Wightman functions is ``beyond the powers 
of perturbation theory'', as Steinmann aptly observes \cite{Steinmann1}.

\subsection{The Singularity Hypothesis}

As remarked in the introduction, one of the most important features of relativistic quantum field theory 
is the behaviour of the theory at large momenta (or large energies). Renormalization group theory \cite{Weinb3}
has contributed a significant lore to this issue (even if none of it has been made entirely rigorous): see, in particular,
the paper of Symanzik on the small-distance behavior analysis of the two-point functions in relativistic 
quantum field theory \cite{Sym}. It strongly suggests that the light-cone singularity of the two-point functions
of interacting theories is stronger than that of a free theory: this is expected even in asymptotically free
quantum chromodynamics, where the critical exponents are anomalous. We refer to this as the ``singularity hypothesis'',
which will be precisely stated in the next section.

\subsubsection{Steinmann Scaling Degree and a theorem}

In order to formulate the singularity hypothesis in rigorous terms, we recall the 
Steinmann scaling degree $sd$ of a distribution \cite{Steinmann}; for 
a distribution $u \in {\cal S}^{'}(\mathbb{R}^{n})$, let $u_{\lambda}$ denote 
the ``scaled distribution'', defined by 
	\[
		u_{\lambda}(f) \equiv \lambda^{-n} u(f(\lambda^{-1} \cdot)) \; . 
	\]
As $\lambda \to 0$, we expect that $u_{\lambda} \approx \lambda^{-\omega}$ for 
some $\omega$, the ``degree of singularity'' of the distribution $u$. Hence, we set
	\begin{equation}
		sd(u) \equiv \inf \, \bigl\{\omega \in \mathbb{R} \mid \lim_{\lambda \to 0} 
		\lambda^{\omega} u_{\lambda} = 0 \bigr\} \; , 
		\label{(24.1)}
	\end{equation}
with the proviso that if there is no $\omega$ satisfying the limiting condition above, 
we set $sd(u) = \infty$. For the free scalar field of mass $m \ge 0$ in four-dimensional
space-time, it is 
straightforward to show from the explicit form of the two-point function in terms of 
modified Bessel functions that
	\begin{equation}
		sd(\Delta_{+}) = 2 \; . 
		\label{(24.2)}
	\end{equation}
In \eqref{(24.2)}, and the forthcoming equations, we omit the mass superscript. 
From Theorem~\ref{th:2}, we have that for $f \in {\cal S}(\mathbb{R}^{4})$ the 
interacting two-point function satisfies
        \begin{align}
		W_{+}(f) = \int_{0}^{\infty}{\rm d}\rho(m_\circ^{2})\int_{\mathbb{R}^{3}}  
		\frac{d\vec{p}}{\sqrt{{\vec{p}\,}^{2}+m_\circ^{2}}} \; 
		\tilde{f} \left(\sqrt{{\vec{p}\, }^{2}+m_\circ^{2}},\vec{p}\, ) \right) \; . 
		\label{(24.3)}
	\end{align}
Here $\tilde{f} \in {\cal S}(\mathbb{R}^{4})$ denotes the Fourier transform of $f$.  

\begin{definition}
\label{def:1} 
We say that the \emph{singularity hypothesis} holds for an interacting scalar field if 
	\begin{equation}
		sd(W_{+}) > 2 \; . 
		\label{(24.4)}
	\end{equation}
\end{definition}

Further support for Definition 3.3 comes from the fact that the singularity hypothesis is indeed 
satisfied at finite orders of perturbation theory (if the interaction density has engineering dimension
larger than 2).

\begin{theorem}
\label{thm:3} If the total spectral mass is finite, \emph{i.e.},
	\begin{equation}
		\int_{0}^{\infty} {\rm d}\rho(a^{2}) < \infty \; , 
		\label{(24.5)}
	\end{equation}
then
	\begin{equation}
		sd(W_{+}) \le 2 \; ; 
		\label{(24.6)}
	\end{equation}
\emph{i.e.}, the scaling degree of $W_{+}$ cannot be strictly greater than 
that of a free theory, and thus, by Definition~\ref{def:1}, the singularity 
hypothesis~\eqref{(24.4)} is \emph{not} satisfied.
\end{theorem}

\begin{proof} 
The scaled distribution corresponding to $W_{+}$ is given by
	\begin{align}
		& W_{+,\lambda}(f) = \lambda^{-2}\int_{0}^{\infty}
		{\rm d}\rho(m_\circ^{2}) \int_{\mathbb{R}^{3}} 
		\frac{d\vec{p}}{\sqrt{{\vec{p}\,}^{2}+\lambda^{2}m_\circ^{2}}} \, 
		\tilde{f} \left( \sqrt{{\vec{p}\,}^{2}
		+\lambda^{2}m_\circ^{2}},\vec{p}\, \right) \; . 
		\label{(24.7)}
	\end{align} 
Assume the contrary to \eqref{(24.6)}, \emph{i.e.}, that $sd(W_{+})=\omega_{0}>2$. 
Then, by the definition of the $sd$, if $\omega < \omega_{0}$, one must have
	\begin{equation}
		\lim_{\lambda \to 0} \lambda^{\omega} W_{+,\lambda}(f) \ne 0 \; . 
		\label{(24.8)}
	\end{equation}
Choosing 
	\begin{equation}
		\omega = \omega_{0}-\delta >2
		\label{(24.9)}
	\end{equation}
in \eqref{(24.8)}, we obtain from \eqref{(24.7)} and \eqref{(24.8)} that
	\begin{align}
		\lim_{\lambda \to 0} \lambda^{\omega-2} 
		\left( \int_{0}^{\infty}{\rm d}\rho(m_\circ^{2}) \int_{\mathbb{R}^{3}}
		 \frac{d\vec{p} }{\sqrt{{\vec{p}\,}^{2}+\lambda^{2}m_\circ^{2}}}
		 \tilde{f}\left(\sqrt{{\vec{p}\,}^{2}+\lambda^{2}m_\circ^{2}},\vec{p}\, \right) 
		 \right)\ne 0 \; . 
		\label{(24.10)}
	\end{align}
The limit, as $\lambda \to 0$, of the term inside the brackets in \eqref{(24.10)}, 
is readily seen to be finite by the Lebesgue dominated convergence 
theorem due to the assumption \eqref{(24.5)} and the fact that 
$\tilde{f} \in {\cal S}(\mathbb{R}^{4})$; but this contradicts \eqref{(24.8)} 
because of \eqref{(24.9)}.
\end{proof}

\begin{corollary}
\label{cor:1}
The singularity hypothesis holds for an interacting scalar field only if 
$\int_{0}^{\infty} d\sigma(m_\circ^{2}) = \infty$. This necessary condition is independent
of the value of $0 \le Z < \infty$.
\end{corollary} 

One importance of the above theorem, which is our main result, 
and especially of its corollary
is that it provides, as we shall see next, a mathematical foundation for the forthcoming 
interpretation of the condition
        \begin{equation}
		Z=0 \; . 
		\label{(24.11)}
	\end{equation}

\subsubsection{The ETCR hypothesis and its consequences for the singularity hypothesis}

For the purposes of identification with Lagrangian field theory, one may equate the $A(.)$ of
\eqref{(9)} with the ``bare'' scalar field $\phi_{B}$ (\cite{Weinb1}, p. 439, see also
\cite{IZ} and \cite{Haag}, whereby
        \begin{equation}
		A = \sqrt{Z} A_{phys}
		\label{(24.12)}
	\end{equation}
under the condition \eqref{(14)}. Under the same condition \eqref{(14)}, the assumption of 
equal time commutation relations (ETCR) 
for the physical fields may be written (in the distributional sense)
		\begin{equation}
			\left[ \frac{\partial A_{phys}(x_{0},\vec{x}\,)}{\partial x_{0}},
			A_{phys}(x_{0},\vec{y}\,) \right] = -\frac{i}{Z} \; \delta(\vec{x}-\vec{y}\,) \; . 
			\label{(24.13)}
		\end{equation}
Together with \eqref{(9)} and \eqref{(24.12)}, \eqref{(24.13)} yields (\cite{Weinb1}, (10.7.18)):
	\begin{equation}
		1 = \int_{0}^{\infty} {\rm d}\rho(m_\circ^{2}) \; . 
		\label{(25)}
	\end{equation} 

Since $d\sigma$ in \eqref{(13)} is a positive measure, we obtain from \eqref{(25)} the inequality 
	\begin{equation}
		Z \le 1
		\label{(26.1)}
	\end{equation}
(\cite{Weinb1}, p.~361). We thus find, comparing \eqref{(25)} with corollary ~\ref{cor:1}:

\begin{corollary}
\label{cor:2}
The singularity hypothesis is incompatible with the ETCR hypothesis. Thus, \eqref{(26.1)} 
is not generally valid.
\end{corollary}

In perturbation theory, 
$Z_{3}(\Lambda)$ (\cite{Weinb1}, p.~462) satisfies \eqref{(26.1)} 
for all ultraviolet 
cutoffs~$\Lambda$, but it is just this condition which relies on the ETCR assumption 
and is not expected to be generally valid. In the limit $\Lambda \to \infty$, however, 
$Z_{3}(\Lambda)$  tends to~$- \infty$ and hence
violates \eqref{(14)} maximally.  In fact, 
\eqref{(14)} is violated  even for finite, sufficiently large $\Lambda$.

Although Corollary ~\ref{cor:2} is strong, in that it excludes the ETCR entirely, together with
its consequences, it trivially implies that \emph{all} values of
$Z$ compatible with \eqref{(14.1)} are allowed, without singling out the condition $Z=0$.

In this connection, we should remark that the ``bare'' (conventional) field \eqref{(24.12)}
in \eqref{(9)} should not be identified with the Wightman field, because the latter is supposed
to have a particle interpretation, through
scattering theory (see Remark~\ref{Remark 0.2}). This is
due to the one-particle space normalization (see also Jost's monograph \cite{Jost}, (1), p.120
for a detailed discussion of this point in connection with the Haag-Ruelle theory). We are thus 
led to replace the $A$ on the l.h.s.~ of \eqref{(9)} by the physical, renormalized field \eqref{(24.13)}.
When \eqref{(14)} holds, the resulting spectral measure is, of
course, $dg(m_{\circ}^{2}) = \frac{1}{Z} d\rho(m_{\circ}^{2})$ which satisfies, by \eqref{(25)},
\begin{equation}
\frac{1}{Z} = \int dg(m_{\circ}^{2})
\label{(26.2)}
\end{equation}
\eqref{(26.2)} is Wightman's (17) in his famous paper \cite{WightPR}. The same
formula is found in Lehmann \cite{Lehmann}, K\"{a}ll\'{e}n \cite{Kallen}, and Barton \cite{Barton1}. Of 
course, Theorem 3.2 remains valid as long as \eqref{(14)} holds. As a consequence, we have the following result:

\begin{corollary}
\label{cor:3}
If, in \eqref{(9)}, the field $A$ is assumed to be the physical field \eqref{(24.12)}, and the ETCR \eqref{(24.13)} is assumed
as well, only the value $Z=0$ remains in \eqref{(14.1)} as possibly compatible with the singularity hypothesis.
\end{corollary}  

It is \eqref{(26.2)} which seems to underlie the general belief that  $Z=0$ is generally
true in an interacting theory, a conjecture which has been made even, for instance, by the 
great founders of axiomatic (or general) quantum 
field theory, Wightman and Haag. Indeed, 
in \cite{Wight}, p.~201, it is observed that ``$\int_{0}^{\infty} {\rm d}\rho(m_\circ^{2}) = \infty$ 
is what is usually meant by the statement that the field-strength renormalization 
is infinite''. This follows from \eqref{(26.2)}, with ``field-strength renormalization'' interpreted 
as $\frac{1}{Z}$. The connection with the singularity hypothesis comes 
next (\cite{Wight}, p.~201), with the observation that, by \eqref{(9)}, $W_{2}$ will 
have the same singularity, as $(x-y)^{2}=0$, as does $\Delta_{+}(x-y;m^{2})$. 
As for Haag, he remarks (\cite{Haag}, p.~55):``In the renormalized perturbation 
expansion one relates formally 
the true field $A_{phys}$ to the canonical field $A$ (our notation) which satisfies 
\eqref{(24.13)}, where $Z$ is a constant (in fact, zero). This means that the fields 
in an interacting theory are more singular objects than in the free theory, and we 
do not have the ETCR.'' 

Although both assertions above seem to substantiate the conjecture that $Z=0$ is expected to 
be a \emph{general} condition for interacting fields, it happens, however, that the resulting ``$\infty \times \delta$'' 
behavior of \eqref{(24.13)}, while suggestive, is entirely misleading. In other, equivalent, words, the
relation \eqref{(26.2)} is rigorously justified if $\int dg(m_{\circ}^{2})<\infty$ as shown by Wightman in 
\cite{WightPR}, p. 863, but in the opposite case $\int dg(m_{\circ}^{2})=\infty$, as required by Theorem 3.2,
\eqref{(26.2)} is misleading, as the example in (\cite{WightPR}, p. 863) demonstrates.  

It follows from Corollary~\ref{cor:2} and Corollary~\ref{cor:3} that 
for both choices of fields in \eqref{(9)}, $Z=0$ is not generally true. The latter case
hinges on the fact that the ETCR is not
generally valid for interacting fields, as briefly reviewed in the 
forthcoming paragraph. We conclude
that the singularity hypothesis opens the possibility of the 
non-universal validity of the equation $Z=0$.

The hypothesis of ETCR has been in serious doubt for a long time, see, 
\emph{e.g.}, the remarks in \cite{StreWight}, p.~101. Its validity has been 
tested \cite{WFW} in a large class of models in two-dimensional space-time; 
the Thirring model \cite{Thirrm}, the Schroer model \cite{Schrm}, the 
Thirring-Wess model of vector mesons interacting with zero-mass 
fermions (see \cite{ThirrWessm}, \cite{DubTar}), and the Schwinger model \cite{Schw}, 
using, throughout, the formulation of Klaiber \cite{Klai} for the Thirring model, 
and its extension to the other models by Lowenstein and Swieca \cite{LSwi} - 
for the Schwinger model, the previously mentioned noncovariant gauge ``$\alpha = \sqrt{\pi}$'' 
was adopted. Except for the Schwinger model, whose special canonical structure is due to 
its equivalence (in an irreducible sector) to a theory of a free scalar field of positive mass, 
the quantity
	\begin{equation}
		\{ \psi(x),\psi(y) \} - \langle \Omega,\{\psi(x),\psi(y)\}\Omega \rangle   
		 \cdot \mathbf{1} \; , 
	\end{equation}
where the $\psi$'s are the interacting fermi fields in the models
and $\{ \, . \, , \, . \,  \}$ denotes 
the anti-commutator and $\Omega$ denotes the vacuum, do not exist in the equal 
time limit as operator-valued distributions, for a certain range of coupling constants. 
Two different definitions of the equal time limit were used and compared, one of them 
due to Schroer and Stichel \cite{SchrSti}.
The models also provide examples of the validity of the singularity hypothesis (for the currents, 
analogous assertions hold if the commutator is used in place of the anti-commutator).
Thus, the ETCR is definitely not true in general.

Although (see Remark~\ref{Remark 0.2}), when $0<Z< \infty$, $Z$ is interpreted as
the non-perturbative field strength renormalization, relating ``bare'' fields to physical
fields, as in \eqref{(24.13)}, the remaining case \eqref{(24.11)} remains to be understood.
As stated in (\cite{Weinb1}, pg. 461), ``the limit $Z=0$ has an interesting interpretation
as a condition for a particle to be composite rather than elementary''. This brings us to our next topic.

\section{A proposal for the meaning of the condition $Z=0$: the presence of massless and unstable particles} 

Buchholz \cite{Buch} used Gauss' law to show that the discrete spectrum 
of the mass operator
	\begin{equation}
		P_{\sigma} P^{\sigma} = M^{2} = P_{0}^{2}-\vec{P}^{2}
		\label{(27)}
	\end{equation}
in a charged sector is empty. Above, $P^{0}$ is the generator of time translations in the physical 
representation, \emph{i.e.}, the physical hamiltonian $H$, and $\vec{P}$ is the physical 
momentum. This fact is interpreted as a confirmation of the phenomenon that, in a
charged sector, and given certain interactions (satisfying Gauss' law) between massless 
particles and others carrying an electric charge, the latter are converted to infraparticles
and are accompanied by clouds of soft photons.

Buchholz formulates adequate assumptions which must be valid in order that one may 
determine the electric charge of a physical state $\Phi$ with the help of Gauss' law 
		\begin{equation}
			e\langle \Phi, j_{\mu} \Phi \rangle= \;
                        \langle \Phi, \partial^{\nu}F_{\nu,\mu} \Phi \rangle \; . 
			\label{(28)}
		\end{equation}
\eqref{(28)} is assumed to hold in the sense of distributions on ${\cal S}(\mathbb{R}^{n})$.

When endeavouring to apply Buchholz's theorem to concrete models such as $qed_{1+3}$, 
problems similar to those occurring in connection with the charge superselection 
rule \cite{StrWight} arise. The most obvious one is that Gauss' law 
\eqref{(28)} is only expected to be 
valid (as an operator equation in the distributional sense) in non-covariant 
gauges (see \eqref{(8)}). Recalling \eqref{(20.3)}, we find that, in an interacting theory 
satisfying the axioms a.) to f.) for
observable fields, in which massless particles (photons), as well as infraparticles
(``electrons with their photon clouds'') occur in a charged sector - that is, in quantum
electrodynamics in Buchholz's characterization \cite{Buch}-, we have:

\begin{corollary}
\label{cor:4} 
        \begin{equation}
		Z_{2} = 0 \; . 
		\label{(29)}
	\end{equation}
\end{corollary}
  
It is interesting to recall, in connection with Corollary ~\ref{cor:4}, that in \cite{JaWre2}, 
\emph{both} the photon field  and the electron-positron field are ``dressed''. The photon is,
however, believed to be a particle \cite{Duch}, and it might be expected that the ``clouds''
around the photon vanish asymptotically. Indeed, the condition $Z_{3}=0$ does not characterize
massless particles (photons), see \cite{BuchCTMB}; only \eqref{(19.3)} remains true.

We now return to the subject 
of \emph{composite} or \emph{unstable particles}.  Classically speaking, these are 
particles whose field does not appear in the 
Lagrangian (\cite{Weinb1}, p.~461, \cite{Weinb4}). We were, however, unable to render this notion precise
even in the classical case (and we thank the referee for convincing us of this fact).

There are, as yet, no rigorous results on the quantum theory, except
for cutoff theories, see \cite{Barbaroux} and references given there. Since
the very notion of particle involves the Poincar\'{e} group, there are
serious  difficulties with the concept of unstable particle in a theory
where the cutoffs are not removed. For a model of Galilei invariant
molecular dynamics with particle production, see \cite{HoJa}.

The condition $Z=0$ appears, however, in 
various non-rigorous approaches to the concept 
of unstable particle \cite{Luk}, \cite{Velt}, \cite{Licht}. 
Turning to scalar fields for simplicity, we consider the case of a scalar particle $C$, 
of mass $m_{C}$, which may decay into a set of two (for simplicity) stable particles, each 
of mass $m$. We have energy conservation in the rest frame of $C$,  \emph{i.e.}, 
	\[
		m_{C}=\sum_{i=1}^{2} \sqrt{{\vec{q}_{i}\,}^{2}+m_{i}^{2}} 
		\ge \sum_{i=1}^{2} m_{i} \; , 
	\]
with $m_{i}=m,i=1,2$, and $\vec{q}_{i}$  the 
momenta of the two particles in the rest frame of $C$:
	\begin{equation}
		m_{C} > 2m  \; . 
		\label{(30)}
        \end{equation}

In order to check that 
        \begin{equation}
           Z_{C} = 0
              \label{(31)}
        \end{equation} 
when \eqref{(30)} holds, while $0<Z_{C}<\infty$ is valid
in the stable case $m_{C} < 2m$, in a model, we are beset with the difficulty to obtain information 
on the two-point function. An exception are those rare cases
in which the (Fock) zero particle state is persistent (\cite{HeppB}),\emph{i.e.}, Lee-type models.

The quantum model of Lee type of a composite (unstable) particle, 
satisfying \eqref{(30)}, where \eqref{(31)}
was indeed found, is that of Araki et al.~\cite{AMKG}. Unfortunately, however, the (heuristic) results
in \cite{AMKG} have one major defect: their model contains ``ghosts''. The paper \cite{HouJou}, cited
by Weinberg, assumes, however, the inequality opposite to \eqref{(30)}, and so is concerned with
stable particles. In fact, their reference to \eqref{(30)} treats it as the \emph{general} conjecture
based on \cite{WightPR} and \cite{Kallen} previously referred to.

There exists, however, a ghostless version of
of the model treated heuristically in \cite{AMKG}, with the correct kinematics, due to 
Hepp (Theorem 3.4, p.~54, of \cite{HeppB}). In this version, the masses in \eqref{(30)},
which are, of course, renormalized masses, may be determined
rigorously from the selfadjointness of the renormalized Hamiltonian.
It is an open problem of great importance to carry out this investigation in detail: 
it would be the first
rigorous model of unstable (composite) particles in quantum field theory.

For atomic resonances, the model in \cite{AMKG} may be treated rigorously, see \cite{Wreszunst}.

\section{A proposal for a criterion for quark confinement in quantum 
chromodynamics (qcd)}

The fundamental objects of quantum chromodynamics (qcd) (\cite{Weinb3}, 
Section 18.7) are the
color gauge-covariant field strength tensor $F_{a}^{\mu\nu}$, where $a=1,2,3$ 
is the colour index and
$\mu,\nu$ are Minkowski space indexes, and the quark field $\Psi_{a}$. 
The color gauge vector
potential $A_{a}^{\mu}$ describes massless gluons. The unobservability of 
quarks has never been explained,
but, in the seventies and eighties, several models, notably the Schwinger 
model \cite{Schw}, \emph{i.e.}, the 
electrodynamics of massless electrons in two spacetime dimensions, 
have been suggested as models of the 
confinement of quarks. In this model, the rigorous version of the 
observable \eqref{(8a)} (\cite{LSwi},
Section IV) is a bilocal quantity which creates a charge dipole with 
an electric field in between, according
to Gauss' law: this picture has been analysed in great detail by 
Casher, Kogut and Susskind \cite{CaKoSu}.
The electron field becomes a functional of the photon field, which acquires a positive mass: let the corresponding 
field at time zero be denoted by $\phi$. If 
$\Omega$ is the Fock vacuum, we define the dipole state (corresponding to placing the ``string'' at time
$t=0$; we omit the time variable for simplicity):
        \begin{equation}                      
        \label{(32)}
                \Psi_{R,\epsilon} = T_{R,\epsilon} \Omega
        \end{equation}
where
        \begin{equation}
        \label{(33)}
                T_{R,\epsilon}=\exp[i  \sqrt{\pi} \, \pi(g_{R,\epsilon})] \; . 
        \end{equation}
Above $\pi$ denotes the field momentum operator at zero time, 
conjugate to $\phi$ and $g_{R,\epsilon}$ is the function
        \begin{equation}
        \label{(41)}
                g_{R,\epsilon}(x) = \begin{cases}
                                   1 & \mbox{if $0\le x\le R$}\, ;\\
                                   0 & \mbox{if $x \le -\epsilon$ or $x \ge R+\epsilon$} \, .
                                   \end{cases}
        \end{equation}                                       
The above definitions follow \cite{LSwi}, pp.~183--185, to which we refer
the reader for more details, except that we introduced $\epsilon>0$ fixed 
and let $R \to \infty$, for reasons of rigor.
The connection between $\phi$ and the electric field (electromagnetic 
field tensor $F_{01}$) at time zero is thereby given as
        \begin{equation}
        \label{(42)}
               \phi(x) = F_{01}(x) \; . 
        \end{equation} 
As in (\cite{BarW}, Appendix A, where in (A.1) the mass term 
should be replaced by $ \frac{e^{2}}{\pi}\phi(x,0)^{2}$ 
inside the Wick dots), we consider the family of states
        \begin{equation}
        \label{43)}
               \omega_{R,\epsilon}(\, \cdot\, ) 
               = \langle \Psi_{R,\epsilon}, \, \cdot \, \Psi_{R,\epsilon} \rangle 
        \end{equation}
on the Weyl CCR algebra generated by $\exp(i \phi(f))$ and $\exp(i \pi(g))$, 
$f,g \in {\cal S}^{0}_{\mathbb{R}}(\mathbb{R})$, the
Schwartz space of real valued functions on the real line, whose Fourier transform 
vanishes at the origin. By compactness (Theorem 2.3.15 of \cite{BRo1}) there exists at least
one limit in the weak* topology:
        \begin{equation}
        \label{(44)}
               \omega_{\epsilon} =\lim_{R \to \infty} \omega_{R,\epsilon} \; . 
        \end{equation}    
The generator $H$ of time translations in the Fock (vacuum) representation is given by
        \begin{equation}
        \label{(45)}
               H = \frac{1}{2} \int {\rm d} x \; {:} [\pi^{2}(x) + 
               (\nabla \phi)^{2}(x)  
               + m^{2} \phi^{2}(x)] {:} \; . 
        \end{equation}
Above, the dots indicate Wick ordering and 
        \begin{equation}
        \label{(46)}
               m = \frac{e}{\sqrt{\pi}}
        \end{equation}
is the dynamically generated photon mass, whose origin is 
topological (see the preface of \cite{FJ} and \cite{SwiecaJ1}, p.~317).
Let $V(g)=\exp(i\pi(g))$. Using 
	\[
		:\phi(x)^{2}: = \lim_{x_{1},x_{2} \to x} [\phi(x_{1}) \phi(x_{2})
		- \langle \Omega, \phi(x_{1})\phi)(x_{2}) \Omega \rangle \mathbf{1}] 
	\]
and similarly for $:\nabla(\phi)^{2}(x):$, we obtain
	\[
		V(h)^{-1} H(x) V(h) = H(x) - 
		\nabla h(x) \cdot \nabla \phi(x) 
		+ \frac{1}{2} (\nabla h)^{2} (x)
		+ \frac{e^{2}}{\pi}\Bigl( h^{2}(x) -2h(x)\phi(x) \Bigr)
	\]
from which a.)~of the following proposition follows:
\begin{proposition}
\label{prop:4}
\quad 
\begin{itemize}
\item [$a.)$] 
	\[
		\langle \Psi_{R,\epsilon}, H \Psi_{R,\epsilon} \rangle = 
		\frac{\pi}{2}\int_{-\infty}^{+\infty} {\rm d} x \; 
		\Bigl(g_{R,\epsilon}^{'})^{2}(x)
		+\frac{e^{2}}{\pi}g_{R,\epsilon}^{2} (x) \Bigr) \; ; 
	\]
where the prime above denotes the derivative;
\item [$b.)$] Any state $\omega_{\epsilon}$, defined by \eqref{(44)}, has unit charge. 
\end{itemize}
\begin{proof} 
See Proposition~A.1 of Appendix A of \cite{BarW} for the 
definition of charge and the proof of item b.). 
\end{proof}
\end{proposition}

The above proposition suggests the following definition:

\begin{definition}
\label{def:2}
We say that $\omega_{\epsilon}$ has \emph{finite energy} if $\Psi_{R,\epsilon}$ lies in the quadratic form domain of $H$
for any $R < \infty$ and $\epsilon > 0$, and
         \begin{equation}
         \label{(47)}
             \limsup_{R \to \infty} \;  \langle \Psi_{R,\epsilon}, 
             H \Psi_{R,\epsilon}  \rangle < \infty \; . 
         \end{equation}
If $\omega_{\epsilon}$ does not have finite energy (with respect to the vacuum, 
which has by definition energy zero),
\emph{i.e.}, if
         \begin{equation}
         \label{(48)}
             \limsup_{R \to \infty} \; 
             \langle \Psi_{R,\epsilon}, H \Psi_{R,\epsilon} \rangle = \infty \; , 
         \end{equation}
we say that the associated fermions are \emph{confined}.
\end{definition}

By \cite{BarW}, Proposition A.2 of Appendix A, the representation 
defined by $\omega_{\epsilon}$ is inequivalent
to the Fock (vacuum) representation. This result follows from the 
fact that the state corresponds to a nonzero
eigenvalue, namely one, of the charge operator, but is not due 
to the existence of a macroscopic energy barrier
between $\omega_{\epsilon}$ and the vacuum state. Indeed, 
when charged particles exist, a system containing a
finite number of these charged particles should have finite 
energy relative to the vacuum, 
which is normalized to zero: this is the physical meaning 
of \eqref{(47)}; this condition is satisfied in
the Streater-Wilde model of soliton sectors, where the analogue 
of the l.h.s.~of~\eqref{(47)} is seen to be identical
to the energy of classical soliton solutions of the two-dimensional 
wave equation \cite{StrWi}. Note, however,
that in their model the scalar field is massless. In our case we have the 
following corollary of Proposition~\ref{prop:4}:

\begin{corollary}
\label{cor:5}
In the Schwinger model, Definition~\ref{def:2} implies that the 
fermions are confined.
\begin{proof} 
The first term in the r.h.s.~of a.)~of Proposition~\ref{prop:4} is 
uniformly bounded in $R$, the second is, by \eqref{(41)},
        \begin{equation}
        \label{(49)}
               E_{d} \equiv \int_{-\infty}^{\infty} {\rm d} x \; \langle \Psi_{R,\epsilon}, 
               :F_{01}(x)^{2}: \Psi_{R,\epsilon} \rangle \geq \mbox{ const.} \; R.
        \end{equation}     
\eqref{(49)} implies \eqref{(48)} and thus the fermions are confined.
\end{proof}
\end{corollary}

Corollary~\ref{cor:5} provides the following physical interpretation 
for the mass term in the Schwinger model 
(see also \cite{CaKoSu}) as $R \to \infty$: it represents the energy 
of a dipole in the non-relativistic limit,
since, in one dimension, the Coulomb potential is linear. In general, 
we have the following distributional formula
(\cite{Guelfand}, p.~361):
       \begin{equation}
       \label{(50)}
              {\cal F}(|\vec{p}|^{\lambda}) 
              = \mbox{ const. } |\vec{x}|^{-\lambda-s} \; , 
       \end{equation}
where ${\cal F}$ denotes the distributional Fourier transform, 
and $s$ the space dimension. 

For $Re \lambda \le -s$ the function $|p|^{\lambda}$ is not locally integrable,
but it defines a distribution in the sense of Guelfand and Chilov (see
\cite{Guelfand}, p. 71).

If $\lambda = -2$ and $s=1$,
we have the present case. If $s=3$, we find $R^{-1}$ for large $R$ 
as expected. In the non-abelian case, we would have in
\eqref{(49)}, $:F_{\mu\nu}^{a}(x)F^{\mu\nu}_{a}(x):$ instead of $:F_{01}(x)^{2}:$,
where the summation over the color index a is understood. 
Due to the $A \wedge A$ term in the
gluon field tensor, we have in \eqref{(50)} the leading infrared singularity 
with $\lambda = -4$ as $|\vec{p}| \to 0$,
which implies the same linear behavior as \eqref{(49)}! Thus, if observable 
fields of type \eqref{(7)} may be constructed
for qcd, whereby $f$ and $g$ would have compact support around 
points growing linearly with a parameter $R$ along 
a radial direction, yielding a quark-antiquark pair, and if the present 
analogy is sound, one expects confinement.
Definition~\ref{def:2} would have, however, to allow for a ``path'' dependence 
of the l.h.s.~of \eqref{(48)} 
and require the validity of \eqref{(48)} independently of the ``path''.
   
It should be remarked that in qcd it is not the fermion number which is 
expected to be confined, but color, which is a multiplicative
charge. A very interesting soluble abelian model of triality (``charm'', an 
abelian version of the three color states
of a free quark) is found in Casher, Kogut and Susskind \cite{CaKoSu}: 
all states of nonzero ``charm'' are confined.

It was Casher, Kogut and Susskind \cite{CaKoSu} who first proposed 
that the deep inelastic structure functions of the theory
might have the scaling laws of the underlying Fermi fields in the Schwinger 
model, although only massive Bosons appear
as asymptotic states. This was proved by Swieca: the short-distance 
limit of the n-point functions of
the observables in the model are those of a free theory of charged 
massless fermions and massless photons (\cite{SwiecaJ1},
p.~317).

In qcd it was already suggested by Cornwall in 1982 \cite{Cornwall} that 
a gluon mass is dynamically generated. 
He used the Dirac or light-cone gauge, in which ghosts are absent. 
Further work in covariant gauges suggests that,
the dynamically generated gluon mass is primarily due to the simplest 
severe infrared singularity
$|\vec{p}|^{-4}$ mentioned above \cite{Gog}, see also \cite{BarGog}. If 
this were so, the analogy between qcd and the Schwinger 
model would be complete: the short-distance limit would imply the 
vanishing of the dynamically generated gluon mass, and, with it, 
the quarks would reappear as free particles, together with massless 
gluons. Note that due to \eqref{(46)}, 
there is no interaction in the limit when the dynamically generated mass 
of the photon tends to zero: there is ``asymptotic freedom''. 
The scaling dimension of the gluon fields is, however, expected to 
be anomalous, as previously observed.

The gluon mass also fits well in the general conjecture of the mass 
gap in general Yang-Mills theories, see the problem
posed by Jaffe and Witten in~\cite{JafWit}.

\section{Conclusion}

We have suggested a criterion which characterizes interacting theories 
in a proper Wightman framework, based on a ``singularity hypothesis''.
The associated framework relies on the axioms for the Wightman 
n-point functions for observable 
fields, including positivity, and thus requires the use of 
non-covariant gauges. The singularity hypothesis is either
incompatible with the ETCR (Corollary~\ref{cor:2}) or 
requires that the (nonperturbative) wave function renormalization
constant $Z$ equals zero, if the ETCR is assumed (Corollary~\ref{cor:3}), depending on 
which type of field - the ``bare''( conventional), or the physical
(renormalized) field in Lagrangian perturbation theory is identified
with the Wightman field in the theory proposed in \cite{WightPR}.

Since the ETCR is generally not valid for interacting field theories, 
this opens up the possibility
that the condition $Z=0$, assumed to be universal for interacting theories, 
is of specific nature. We propose that it describes ``dressed'' infraparticles (in the presence of 
massless particles) in certain specific theories, such as quantum electrodynamics in 
Buchholz's characterization \cite{Buch}, or of composite (unstable) particles.
 
Since the gluons are massless, ``dressed'' quarks and gluons 
may occur in the theory of strong interactions. 
A proposed caricature of (some aspects of) qcd such as quark 
confinement is the Schwinger model  
revisited in Section 5. There, the ``dressing'' of the electrons assumes 
a drastic form: by Bosonization,
the electron field is a functional of the photon field, which acquires 
a mass. We propose a criterion
of confinement which is valid in this model, and whose extrapolation 
to qcd, if valid, predicts
a surprisingly realistic picture. It is important that this conjectured 
extrapolation depends on
a description in terms of \emph{observables} (compare ref. \cite{Buch1}), 
whose role in the Schwinger model
was emphasized by Lowenstein and Swieca \cite{LSwi}. Accordingly, in qcd, we also 
expect the analogue(s) 
of the condition $Z=0$ due to the massless gluons; the Schwinger model is, 
however, canonical, as explicitly 
verified in \cite{WFW}. This is due to the fact that the ``dressing''of the electrons
by the photons is such that, apart from the $\theta$ angle characterizing the irreducible
representation, they entirely disappear from the picture, originating
a free theory of (massive) photons in Fock space, and thereby accounting for the model's
dynamical solubility.
An example of a canonical interacting field theory, for entirely different 
reasons, is given in 
\cite{JaeMu}. 
  
The proposed framework poses several difficult problems.  The construction
of the ``dressed'' observables remains open even in $qed_{1+3}$, in 
spite of the results in \cite{JaWre2}.
In qcd, the methods of \cite{JaWre2} are not directly applicable, due to 
the gluon self-interaction.

In \cite{SwiecaJ1}, there is a final remark: ``After almost a century of 
existence the main question about
quantum field theory seems still to be: what does it really describe? 
and not yet: does it provide a good
description of nature?''. The fact that all but the lightest particles are 
unstable and there is as yet no
rigorous model in quantum field theory to describe them (see the end 
of Section 4) is a clear instance of the fact
that quantum field theory has lost contact with its prime object 
of study, the elementary particles, and 
therefore with nature itself. As an important subarea of mathematical 
physics, it seems to have moved
in the direction contrary to Geoffrey Sewell's suggestion that ``in the 
words `mathematical physics', `physics' is the
noun and `mathematical' is the adjective''. If this tendency can 
be inverted, it may even be hoped that,
in spite of all the difficulties --- those mentioned above, concerning 
our approach, as well as others
(\cite{JafWit}, \cite{Mund}, \cite{Fredenhagen}) --- 
the question whether we should believe in quantum field theory 
posed by Wightman in the title of \cite{Wight} may be answered in the affirmative.

\section{Acknowledgement}
We are very grateful to the referee. He clarified and corrected several important conceptual issues. We learned 
a lot from his various objections. He convinced us that our chapter on the classical theory of unstable particles
was self-contradictory, and we decided to eliminate it.

\end{document}